\documentclass[%
preprint,
nofootinbib,
amsmath,amssymb,
aps,
superscriptaddress
]{revtex4-1}

\usepackage{amssymb,mathrsfs,physics,amsthm,amsfonts,amsmath,mathtools}
\usepackage{tikz}

\newtheorem{theorem}{Theorem}
\newtheorem*{theorem*}{Theorem}

\newtheorem*{corollary*}{Corollary}

\newtheorem{postulate}{Postulate}

\theoremstyle{definition}
\newtheorem{examplex}{Example}
\newenvironment{example}
  {\pushQED{\qed}\examplex}
  {\popQED\endexamplex}

\usepackage{graphicx}
\graphicspath{ {images/} }

\DeclarePairedDelimiterX{\inp}[2]{\langle}{\rangle}{#1, #2}

\usepackage{xparse}
\newcommand{\Hil}[1]{\mathcal{H}^{#1}}
\NewDocumentCommand\LH{mo}{%
  \IfNoValueTF{#2}
   {\mathcal{L}(\mathcal{H}^{#1})}
   {\mathcal{L}(\mathcal{H}^{#1},\mathcal{H}^{#2})}%
}

\usepackage{cleveref}

\usepackage{upgreek}

\usepackage{CJKutf8}

\begin{document}

\begin{abstract}
There has been a body of works deriving the complex Hilbert space structure of quantum theory from axioms/principles/postulates to deepen our understanding about quantum theory and to reveal ways to go beyond it to resolve foundational issues. Recent progresses in incorporating indefinite causal structure into physical theories suggest that a more comprehensive understanding of both quantum theory and the theory beyond it accounts for indefinite causal structure. We formulate a framework of physical theories without assuming definite causal structure and identify postulates that single out the complex Hilbert space structure. More than one complex Hilbert space theory is compatible with the postulates, which leaves the room for the further search of the best among these theories.
\end{abstract}

\title{Quantum from principles without assuming definite causal structure}

\begin{CJK*}{UTF8}{gbsn}
\author{Ding Jia (贾丁)}
\email{ding.jia@uwaterloo.ca}
\affiliation{Department of Applied Mathematics, University of Waterloo, Waterloo, ON, N2L 3G1, Canada}
\affiliation{Perimeter Institute for Theoretical Physics, Waterloo, ON, N2L 2Y5, Canada}
\maketitle
\end{CJK*}

\section{Introduction}\label{sec:int}

Ordinary quantum theory assumes definite causal structure. This assumption is manifested in the existence of a dynamical law that evolves physical states through a definitely ordered sequence of continuous or discrete times, and in the definite causal order presumed for the quantum operations.

In recent years it is realized that to describe nature more comprehensively it is very likely necessary to drop the assumption of definite causal structure and incorporate indefinite causal structure into the theory. Experiments claiming the realizations of operations with indefinite causal structure had been reported \cite{maclean2017quantum, procopio2015experimental, rubino2017experimental}, and protocols had been discovered offering a further layer of indefinite-causal-structure-over-definite-causal-structure advantage in information processing (e.g. \cite{hardy2009quantum, chiribella2013quantum, chiribella2012perfect, araujo2014computational, feix2015quantum, guerin2016exponential}), in addition to the quantum-over-classical advantage for theories with definite causal structure \cite{nielsen2000quantum}. Moreover, it was pointed out early on that a theory unifying quantum theory and general relativity is expected to have a causal structure that is both dynamical and indefinite \cite{hardy2005probability, hardy2007towards}.

While the pioneer work introduces indefinite causal structure to general operational probabilistic theories \cite{hardy2005probability, hardy2007towards}, more recent works specialize to construct theories and models with the complex Hilbert space structure (e.g., \cite{chiribella2013quantum, oreshkov2012quantum, oreshkov2016operational}). Ordinary quantum theory\footnote{As a note on terminology, we sometimes refer to complex Hilbert space quantum theory as ``quantum theory'' for simplicity. This is to be distinguished from quantum theory on other spaces, such as ``real Hilbert space quantum theory'' \cite{stueckelberg1960quantum}. By ``ordinary quantum theory'' we mean complex Hilbert space quantum theory with definite causal structure.} based on the complex Hilbert space structure suffers foundational problems which motivate people to look for better alternatives \cite{auletta2001foundations}. In particular, there is a body of work that study alternative operational probabilistic theories (see e.g. \cite{chiribella2016quantum} and references therein). Some alternative theories exhibit interesting new features such as larger violations of Bell's inequality than quantum theory \cite{popescu1994quantum, barrett2007information}, but none of the alternatives have so far been found to definitively describe nature better than quantum theory. To answer the deep question of what makes quantum theory special in the landscape of possible probabilistic theories, several different sets of axioms/principles/postulates have been identified which single out quantum theory (e.g. \cite{hardy2001quantum, barnum_higher-order_2014, chiribella2011informational, daki?_quantum_2011, hardy2011reformulating, masanes_derivation_2011, wilce2012four, wilce_conjugates_2012, selby2018reconstructing}). These works usually contain two parts, with the first part offering a framework to formulate a family of probabilistic theories, and the second part deriving that only quantum theory obeys certain postulates. It is hoped that these axiomatic characterizations of complex Hilbert space quantum theory not only tell us what makes quantum theory special, but also guide the continued search for a superior theory that resolves the foundational problems of quantum theory.

The above axiomatic works commonly assume definite causal structure, either at the level of the general framework so that all theories in the landscape have definite causal structure, or at the level of the postulates so that the quantum theory that is singled out has definite causal structure. In view of the need to incorporate indefinite causal structure mentioned at the beginning, the assumption of definite causal structure appears as an important limitation. For the sake of understanding what makes quantum theory special to describe nature so well, it is preferable not to impose definite causal structure because as mentioned above a more comprehensive description of nature likely incorporates indefinite causal structure. For the sake of searching for a theory superior to quantum theory as well, it is preferable not to impose definite causal structure because the superior theory may be a theory with indefinite causal structure.

The task of the present work is to find a set of postulates that singles out the complex Hilbert structure within a framework of theories that does not impose definite causal structure. 

The framework of theories without imposing definite causal structure we use is built on a powerful perspective on physical theories offered by Hardy \cite{hardy2005probability, hardy2007towards}:
\begin{quote}
A physical theory, whatever else it does, must correlate recorded data.
\end{quote}
This motivates us to take \textit{operations} (through which data are recorded) and \textit{correlations} as the basic concepts of the framework, detailed in \Cref{sec:oc}. To give a mathematical structure to the concepts, an important postulate is made so that operations are associated with ordered vector spaces, and correlations are associated with (multi)linear functionals on these spaces. This framework differs from many of the frameworks used in the previous axiomatic works in that correlation, as a concept distinct from operations, plays a very important role.

The task of identifying postulates and deriving the complex Hilbert space structure is made easy by the previous works of Wilce and Barnum \cite{wilce2012four, wilce_conjugates_2012} (see also \cite{barnum2016post} and references therein for a comprehensive account of the approach and \cite{selby2018reconstructing} for a related work based on category theories). The original postulates and derivations in their work are for theories with definite causal structure. Yet we show that the same general strategy of using the Jordan algebra structure to arrive at the complex Hilbert space works in a framework with indefinite causal structure. The list of postulates and the derivation of the complex Hilbert space structure is presented in \Cref{sec:chs}. Some brief concluding remarks are offered in \Cref{sec:con}

\section{Physical theories as theories of operations and correlations}\label{sec:oc}

No matter what ever else a theory of physics does, it must correlate recorded data \cite{hardy2005probability, hardy2007towards}. Data are recorded through operations. There are other things a theory of physics can do, such as categorizing the constituents of the universe and offering a picture of reality, but at a minimum, it must deal with operations and correlations. In this paper, we focus on probabilistic theories. Some basic structures about probabilistic theories taking operations and correlations as fundamental concepts is presented in this section.

\subsection{Operation}
An operation consists of some action and some observation. For example, the game of ``throwing the paper ball into the basket'' involves an operation that consists of the action of picking up the paper ball and throwing it towards the basket, and the observation of seeing whether the paper ball goes into the basket. 

Note that the action and observation do not have to occur in a definite sequence. There are operations with the observation preceding the action, and others with the action and the observation occurring simultaneously. It is helpful to simplify the situation by introducing the notion of ``general action'' to unify action and observation. A general action may be an action with a trivial observation (e.g., Alice throws the paper ball towards the basket and look into the sky without observing whether the ball falls in), a pure observation (e.g., another person Bob observes if Alice's ball falls in), or a combined action-observation (throw the ball and keep on observing where it flies). 

Data is always gathered through the observation part of the general action. The trivial observation with only one possible outcome is still viewed to gather some data, even though this piece of data offers no nontrivial information to distinguish among more than one possibility. 

An operation always refers to some physical objects. In the example above the relevant physical objects are the paper ball and the basket. In general, the relevant physical objects for an operation can be more complicated. For example, the operation of taking an orange and producing a cup of orange juice has the relevant physical object, the orange, going through different forms of existence (raw orange and orange juice). To be specific and talk about the different forms of existence, we speak of the relevant \textit{physical system} of an operation. The physical system shows up as part of the mathematical description of an operation to specify what state of affairs are relevant for the operation. In the example above, we may take the operation to have two relevant physical systems: the state of the orange when it is raw and the state of the orange when it becomes juice. The physical system of an operation specifies a condition that enables the operation and/or a condition that checks the validity of an operation. Only when a paper ball and a basket is present can one play the game of throwing, and only when the orange is turned into juice (but not, say, a half peeled orange) is the operation valid in that context. We note that in some situations the data recorded also invokes physical systems to store the data. For example, in a paper ball throwing competition the result of whether Alice's ball lands in may be recorded on a piece of paper for further reference. This piece of data of either ``yes'' or ``no'' is classical. In other cases the data recorded may take the form of a quantum state or states on some type of systems.

To summarize, in a physical theory, a minimal description of an operation consists of a general action, a set of possible data gathered from the general action, and the relevant physical systems for the operation. More generally, there are situations where multiple choices for the operation are available. A general operation consists of a set of possible general actions, each with its own possible data set and its own relevant physical systems. We settle on this characterization of operations. 

To symbolize an operation we adopt the following convention. A general action is denoted with capital letters in the form $\mathsf{A}$. A physical system is denoted with lower-case letters in the form $\mathsf{a}$. Sometimes we group systems together into a composite system. If the composite physical system $\mathsf{a}$ consists subsystems $\mathsf{a_1},\mathsf{a_2},\cdots, \mathsf{a_n}$, we write $\mathsf{a}=\mathsf{a_1 a_2\cdots a_n}$ and may use either the left side or the right side to refer to the composite system. The set of possible data is enumerated by letters $i$ in a different font. These symbols $\mathsf{A, a}, i$ can be combined to make explicit different pieces of information. For example, a general action $\mathsf{A}$ with system $\mathsf{a}$ is referred to as $\mathsf{A_a}$, and its $i$-th data may be referred to as $\mathsf{A_a}[i]$. 

In this language, an operation $\mathcal{O}$ is described by an indexed set of objects $\{\mathsf{A_a}[i]\}_{\mathsf{A},\mathsf{a},i}$, where it is understood that the sets of possible values $\mathsf{a}$ and $i$ vary according to the choice of general action $\mathsf{A}$. We write
\begin{align}\label{eq:op}
\mathcal{O}=\{\mathsf{A_a}[i]\}_{\mathsf{A},\mathsf{a},i}.
\end{align}

\begin{example}
A familiar example of operation is the quantum instrument used in quantum theory \cite{davies1970operational}. A quantum instrument is a set of completely positive (CP) maps $\{\mathcal{E}[i]\}_i$ from some input state space $L(\mathcal{H}_{\mathsf{a_1}})$ (the space of bounded linear operators on the complex Hilbert space $\mathcal{H}_{\mathsf{a_1}}$) to some output state space $L(\mathcal{H}_{\mathsf{a_2}})$. The set of maps is required to sum up to a completely positive trace preserving map (channel). The quantum instrument describes a general action whose possible observational outcomes are $i$ and whose physical system has two subsystems. The input subsystem $\mathsf{a_1}$ is the one associated with the space $L(\mathcal{H}_{\mathsf{a_1}})$ and the output system $\mathsf{a_2}$ is the one associated with the space $L(\mathcal{H}_{\mathsf{a_2}})$. We write the composite system of the operation as $\mathsf{a=a_1a_2}$. Then the operation takes the form $\{\mathcal{E}_\mathsf{a}[i]\}_{\mathsf{a},i}$, which is a special case of (\ref{eq:op}) with only one choice for the general action.
\end{example}

\subsection{Correlations and probabilistic theories}

The other basic concept of the framework is the correlation. The correlation among data registered from operations may be established through some other operation that interact with the physical systems of the original operations. In some information theory inspired circuit models of operational probabilistic theories this is the only way to establish correlation. Yet it is also possible that the correlation is established not through other operations conducted by agencies. For example, the global states in quantum field theory establish correlations for operations coupled to the field states, but the global state is not supposed to always be prepared by some other operations. Both kinds of correlations, correlations established through and not through operations, can be described in the present framework.

Correlation is a broad term and in general, a theory mentioning the concept of correlation may not refer to probabilities. Yet in this paper we focus on probabilistic theories. In this context the main function of a probabilistic theory is to calculate probabilities for allowed operations to register certain data. In general, the probabilities to be calculated take the form of conditional probabilities. When a conditional probability is well-defined\footnote{See \cite{hardy2005probability, hardy2007towards} for a discussion on the non-triviality of the requirement that the probabilities are well defined.}, a probabilistic theory is expected to offer a method to calculate it.

In general the conditional probabilities are of the form $p(i,j,\cdots,k|\text{cond})\in\mathbb{R}$, where $i,j,\cdots,k$ is a possible set of data to be registered from a set of general actions, and $\text{cond}$ encode the prerequisite conditions for the probability to make sense. The conditions contain the choice of general action for each operation, and further conditions to make the probabilities well-defined. For example, in a circuit model $\text{cond}$ can include the wiring configurations of the devices. In this probabilistic theory setting a \textit{correlation} specifically refers to a map from a set of data to the set of real numbers, offering information on the conditional probabilities. A central theme of any probabilistic theory is to specify the properties of such maps. A natural structure to be imposed is linearity, which forms the topic of the next subsection.

\subsection{Theory structure regarding probabilities}

Conventionally, \textit{absolute probability} are used for probabilities. The conditional probabilities of the form $p(i,j,\cdots,k|\text{cond})\in \mathbb{R}$ obey
\begin{align}
p(i,j,\cdots,k|\text{cond})\ge &0
\\
\sum_{i,j,\cdots,k}p(i,j,\cdots,k|\text{cond})=&1,
\end{align}
where the sum is over possible data to be recorded from the set general actions. These imply
\begin{align}
1\ge p(i,j,\cdots,k|\text{cond})\ge 0.
\end{align}

There is an alternative option of using \textit{probability weights}. The probability weights $w(i,j,\cdots, k|\text{cond})\in \mathbb{R}$ are only required to obey 
\begin{align}
\infty>w(i,j,\cdots,k|\text{cond})\ge &0.
\end{align}
These probability weights are meaningful in comparison with each other, which saves the need for normalization. For any pair $w(i|\text{cond})$ and $w(j|\text{cond})$ of probability weights (Here for simplicity we used one letter $i$ or $j$ to represent a list of observational outcomes.), if $w(j|\text{cond})\ne 0$, then the prediction is that the data $i$ is $r=w(i|\text{cond})/w(j|\text{cond})$ times as likely to be recorded as $j$. If $w(j|\text{cond})= 0$, a comparison of probability weights in terms of the ratio $r=w(i|\text{cond})/w(j|\text{cond})$ should not be made, and physical meaning is that the data $j$ is predicted never to be recorded.

When $0<\sum_{i,j,\cdots,k} w(i,j,\cdots,k|\text{cond})<\infty$, where the sum is over all possible outcome for the set of general actions, normalization can be conducted and the absolute probabilities can be obtained from the relative probabilitie as
\begin{align}\label{eq:gab}
p(i,j,\cdots,k|\text{cond})=\frac{w(i,j,\cdots,k|\text{cond})}{\sum_{i,j,\cdots,k} w(i,j,\cdots,k|\text{cond})}.
\end{align}
The case of $0=\sum_{i,j,\cdots,k} w(i,j,\cdots,k|\text{cond})$ should not appear in a physically meaningful setup, since among all possible outcomes some outcome should happen. In a physically meaningful setup and for finitely many outcomes, $0<\sum_{i,j,\cdots,k} w(i,j,\cdots,k|\text{cond})<\infty$ always holds, and the absolute probabilities can always be obtained from the probability weights. Whereas the absolute probabilities are unique, the probability weights may be rescaled by the same factor without changing the physical content. This means that two theories using probability weights may give physically equivalent predictions even when the exact values for the probability weights of the same outcomes do not agree. The case of a diverging $\sum_{i,j,\cdots,k} w(i,j,\cdots,k|\text{cond})$ may appear when infinitely many outcomes are allowed by a theory. Then one needs to specify a separate rule to convert probability weights to absolute probabilities, if one still wants to do the conversion. As far as the derivation of the complex Hilbert space structure of this paper goes we do not need to worry about this case, since the number of outcomes will be assumed to be finite.

So far we have been talking about operations as an abstract concept without embedding them in a mathematical model. We will now introduce a basic postulate to endow the operations (along with correlations) with some additional mathematical structure. Under this postulate, observational data will become vector spaces elements, and the map of correlations will become (multi)linear functionals over such vector spaces. 

The motivation comes from the probabilistic mixing of general actions. Let $\mathcal{O}=\{\mathsf{A_a}[i]\}_{\mathsf{A,a},i}$ contain $\mathsf{A_a}$ and $\mathsf{B_a}$ as two choices for the general action associated with the same physical system $\mathsf{a}$. Provided both general actions distinguish finitely many possible outcomes, without loss of generality we can suppose they have the same total number of outcomes (adding void outcomes that are never triggered to the general action with the smaller number of outcomes if needed). Suppose a theory predicts $w(i|\text{cond},\mathsf{A})$ should $\mathsf{A}$ be chosen as the general action to be performed, and $w(i|\text{cond},\mathsf{B})$ should $\mathsf{B}$ be chosen as the general action to be performed. Probabilistically mixing $\mathsf{A}$ and $\mathsf{B}$ means performing  $\mathsf{A}$ with probability weight $w_\mathsf{A}$ and $\mathsf{B}$ with probability weight $w_\mathsf{B}$. Under such a mixing $\{\mathsf{A},\mathsf{B};w_\mathsf{A},w_\mathsf{B}\}$ the predictions for the outcomes is expected to be
\begin{align}\label{eq:pma}
w(i|\text{cond},\{\mathsf{A},\mathsf{B};w_\mathsf{A},w_\mathsf{B}\})=w_\mathsf{A}\bar{w}_\mathsf{B} w(i|\text{cond},\mathsf{A})+w_\mathsf{B}\bar{w}_\mathsf{A} w(i|\text{cond},\mathsf{B}),
\end{align}
where $\bar{w}_\mathsf{A}=\sum_i w(i|\text{cond},\mathsf{A})$, and $\bar{w}_\mathsf{B}=\sum_i w(i|\text{cond},\mathsf{B})$. This formula takes the form of a weighted sum of $w(i|\text{cond},\mathsf{A})$ and $w(i|\text{cond},\mathsf{B})$ by the weights $w_\mathsf{A}\bar{w}_\mathsf{B}, w_\mathsf{B}\bar{w}_\mathsf{A}
\in\mathbb{R}$. $\bar{w}_A$ and $\bar{w}_B$ are present to even out initial inequalities of $\sum_i w(i|\text{cond},\mathsf{A})$ and $\sum_i w(i|\text{cond},\mathsf{B})$ that could be present due to the rescaling degeneracy.

Theories in which equation (\ref{eq:pma}) holds has a certain linear structure for the correlation as a map from the outcomes to the probability weights. It suggests that the recorded data on the same physical system be represented as elements in a vector space, with real numbers such as $w_\mathsf{A}\bar{w}_\mathsf{B}$ and $w_\mathsf{B}\bar{w}_\mathsf{A}$ forming the field for the vector space, and the correlations as multilinear maps from these vector spaces to the probability weights. We realize this suggestion as a postulate.
\begin{postulate}[Linearity]\label{ps:l}
Recorded data for general actions with the same relevant physical system are represented as positive cone elements in an ordered vector space with some trivial data as an order unit. Correlations are represented as positive multilinear functionals on such spaces.
\end{postulate}

Here an \textit{ordered vector space} is a real vector space $V$ endowed with a convex cone $V^+$ such that $V^+$ spans $V$, and that $V^+\cap -V^+=\{0\}$. $V^+$ is called the \textit{positive cone} of $V$. An \textit{order unit} of an ordered vector space is an element $u\in V^+$ so that for any $v\in V$, there is an $a>0$ such that $au-v\in V^+$.

The ordered vector space of Postulate \ref{ps:l} is called an \textit{operational space}, and is denoted in the form $\mathfrak{O}_\mathsf{a}$, where $\mathsf{a}$ is the relevant physical system. The dimension of the space is denoted $d_{\mathsf{a}}$. The positive cone is denoted $\mathfrak{O}^+_\mathsf{a}$. It contains the elements that represent physical data. Each $\mathsf{A_a}[i]$ is represented by an element of $\mathfrak{O}_\mathsf{a}^+$. We refer to these vector space elements using the same symbols $\mathsf{A_a}[i]$ for the observational outcomes when no ambiguity arises. When it is clear from the context we often suppress the labels $[i]$ and refer to the vector space elements in the form $\mathsf{A_a}$ for simplicity.

The correlations as positive multilinear functionals on $\mathfrak{O}_\mathsf{a}$, $\mathfrak{O}_\mathsf{b}, \cdots \mathfrak{O}_\mathsf{c}$ are denoted in the form $\mathsf{D^{ab\cdots c}}$ with the physical systems in the superscript to be distinguished from the recorded data with the system in the subscript:
\begin{align}
\mathsf{D^{ab\cdots c}}:  \mathfrak{O}_\mathsf{a}\times \mathfrak{O}_\mathsf{b}\times \cdots\times \mathfrak{O}_\mathsf{c}&\rightarrow \mathbb{R},\nonumber
\\
(\mathsf{A_a}[i],\mathsf{B_b}[j],\cdots, \mathsf{C_c}[k])&\mapsto w(i,j,\cdots, k|\text{cond}).
\end{align}
The vector space generated by the correlations is called a \textit{correlation space} and is denoted $\mathfrak{C}^{\mathsf{ab\cdots c}}$. The dimension of the correlation space is denoted $c_{\mathsf{ab\cdots c}}$.

\begin{example}
An example of an operational probabilistic theory that incorporates indefinite causal structure and uses probability weights is the ``modified Oreshkov-Cerf theory''.

The original Oreshkov-Cerf theory is an operational quantum theory without predefined time \cite{oreshkov2016operational} (See also \cite{oreshkov2015operational}). A main new feature in comparison to ordinary operational quantum theory is that in accordance with the absence of a predefined time, the systems associated with an operation/general action are not separated into input and output subsystems. 

Using the notations of the original paper, an operation/general action $\{M_i^{AB\cdots}\}_{i\in O}$ consists of a set of possible events/outcomes indexed by the data set element $i\in O$. $A,B,\cdots$ are the physical systems associated with the operation, with corresponding Hilbert spaces $\Hil{A},\Hil{B},\cdots$ whose dimensions are $d^A,d^B,\cdots$. The events are represented by positive semidefinite operators $M_i^{AB\cdots}$ on $\Hil{A}\otimes\Hil{B}\otimes\cdots$. 

Operations come in equivalence classes. Two operations $\{M_i^{AB\cdots}\}_{i\in O}$ and  $\{N_i^{AB\cdots}\}_{i\in O}$ that yield the same joint probabilities for all experimental setups (or circuits) belong to the same equivalence class. Similarly events come in equivalence classes. Two events $M_i^{AB\cdots}$ and  $N_i^{AB\cdots}$ coming from different operations that yield the same joint probabilities with other events in all experimental setups (or circuits) belong to the same equivalence class. 

Events/operations in the same equivalence class have operators that differ by a constant factor. One way to avoid this ambiguity is to represent an equivalence class of events by specifying a pair of operators in the form $(M_i^{AB\cdots},\overline{M}^{AB\cdots})$, where $\overline{M}^{AB\cdots}:=\sum_{i\in O}M_i^{AB\cdots}$, and fixing a normalization convention, such as
\begin{align}\label{eq:nc}
\Tr \overline{M}^{AB\cdots}=d^A d^B\cdots.
\end{align}
The null operation $\{O^{AB\cdots}\}$ with trace zero is treated separately as a singular case.

The normalization requirement (\ref{eq:nc}) is weaker than what is usually imposed in ordinary quantum theory. Ordinary quantum theory is time-asymmetric in the sense that measurement outcomes represented by POVM elements sum up to the identity (or more generally, outcomes represented by quantum instrument elements sum up to a channel), but states in a preparation are only required to have their traces sum up to one. In a theory without predefined time this time-asymmetry should be absent, and in the Oreshkov-Cerf theory the time-asymmetry is eliminated by weakening the requirement on outcomes so that only a sum of trace condition (\ref{eq:nc}) is imposed.

The correlation is encoded in the following formula for joint probabilities:
\begin{align}\label{eq:occ}
p(i,j,\cdots|\{M_i^{\cdots}\}_{i\in O},\{N_j^{\cdots}\}_{j\in Q},\cdots;\text{network})
=\frac{\Tr[(M_i^{\cdots}\otimes N_j^{\cdots}
\otimes \cdots )W_{\text{wires}}]}{\Tr[(\overline{M}^{\cdots}\otimes \overline{N}^{\cdots}
\otimes \cdots )W_{\text{wires}}]}.
\end{align}
This is a special case of (\ref{eq:gab}). The condition in the conditional probability specifies the relevant operations and the way they are connected (``network''). The connection can be specified using a graph. The operations are located at the nodes. Each (sub)system of an operation is connected to a (sub)system of another operation with the same dimension using a ``wire'', which is an edge labelled by the system dimension. A wire tells which system interact with which, and is mathematically described as a pure bipartite entangled state $\ketbra{\Upphi}$ whose precise form depends on the symmetry of the system. The operator $W_{\text{wires}}$ is the tensor product of all these wire operators. This is the Oreshkov-Cerf theory in a nutshell. Details on the motivations and discussions about causality can be found in the original article \cite{oreshkov2016operational}.

The theory as presented so far does not fit into the present framework. The map $(M_i^{\cdots}, N_j^{\cdots}, \cdots)\mapsto p(i,j,\cdots|\{M_i^{\cdots}\}_{i\in O},\{N_j^{\cdots}\}_{j\in Q},\cdots;\text{network})$ according to (\ref{eq:occ}) is not multilinear because of the division by $\Tr[(\overline{M}^{\cdots}\otimes \overline{N}^{\cdots}
\otimes \cdots )W_{\text{wires}}]$. To make the map multilinear and fit into the present framework one could use probability weights with the formula
\begin{align}\label{eq:morpw}
w(i,j,\cdots|\{M_i^{\cdots}\}_{i\in O},\{N_j^{\cdots}\}_{j\in Q},\cdots;\text{network})
=\Tr[(M_i^{\cdots}\otimes N_j^{\cdots}
\otimes \cdots )W_{\text{wires}}].
\end{align}
This map $(M_i^{\cdots}, N_j^{\cdots}, \cdots)\mapsto w(i,j,\cdots|\{M_i^{\cdots}\}_{i\in O},\{N_j^{\cdots}\}_{j\in Q},\cdots;\text{network})$ is then multilinear. 

In comparison to in (\ref{eq:occ}), in (\ref{eq:morpw}) the operators with overline no longer show up. By modifying the theory to use probability weights, we depart from describing operations and events in equivalence classes in the form $(M_i^{AB\cdots},\overline{M}^{AB\cdots})$. There is now a constant multiplicative factor ambiguity in the probability weights, since one is allowed to rescale the operators of the events in the same operation by an arbitrary common positive factor. This ambiguity does not affect the physical predictions, since the probability weights are only meaningful in comparison to each other, specifically through taking ratios.

\end{example}

\subsection{Subsystem structures}

As the last part to specify the basic framework for probabilistic theories with operations and correlations, we discuss the subsystem structure for composite physical systems. We assume two very basic properties for the operational spaces of composite systems. A system $\mathsf{a}$ with $d_{\mathsf{a}}=\dim\mathfrak{O}_{\mathsf{a}}=1$ is called a \textit{trivial system}. The space of a trivial system supports only one linearly independent vector, which describes a trivial data. We assume that for a trivial system $\mathsf{a}$, $\mathfrak{O}_{\mathsf{ab}}\cong \mathfrak{O}_{\mathsf{b}}$ as ordered vector spaces for all $\mathsf{b}$. 

The second basic property we assume is that any operational space $\mathfrak{O}_{\mathsf{ab}}$ with two subsystems contain all the product elements while preserving linear independence, i.e., if $\mathsf{A_a}\in \mathfrak{O}_{\mathsf{a}}$ and  $\mathsf{B_b}\in \mathfrak{O}_{\mathsf{b}}$, then there is an element $\mathsf{A_a}\mathsf{B_b}\in \mathfrak{O}_{\mathsf{ab}}$ so that if $\mathsf{A_a}$ and $\mathsf{A'_a}$ are linearly independent in $\mathfrak{O}_{\mathsf{a}}$ and $\mathsf{B_b}$ and $\mathsf{B'_b}$ are linearly independent in $\mathfrak{O}_{\mathsf{b}}$, then $\mathsf{A_aB_b}$, $\mathsf{A'_aB_b}$, $\mathsf{A_aB'_b}$ and $\mathsf{A'_aB'_b}$ are all linearly independent in $\mathfrak{O}_{\mathsf{ab}}$. This implies that $d_\mathsf{a} d_\mathsf{b}\le d_\mathsf{ab}$. 

There is a similar basic property we assume for the correlations that pertain to two operational spaces. Suppose $\mathsf{C^a}$ is a correlation pertaining to $\mathfrak{O}_{\mathsf{a}}$ itself and $\mathsf{D^b}$ is a correlation pertaining to $\mathfrak{O}_{\mathsf{b}}$. Then we assume that there is a correlation $\mathsf{C^a}\mathsf{D^b}$ pertaining to $\mathfrak{O}_{\mathsf{ab}}$ so that $\mathsf{C^a}\mathsf{D^b}(\mathsf{A_a}\mathsf{B_b})=\mathsf{C^a}(\mathsf{A_a})\mathsf{D^b}(\mathsf{B_b})$, i.e., the probability weights multiply.

\subsection{Comments on the framework}

The framework just presented family-resemble other frameworks used in previous axiomatic works, but have some notable differences. First of all no assumption of definite causal structure is imposed on the current framework. Moreover, correlations carrying non-trivial physical information but not generated by operations is allowed in the current framework. This is in contrast with the circuit models \cite{chiribella2010probabilistic, chiribella2011informational, hardy2011reformulating}, where the operations carry non-trivial physical correlation and the ``wires'' do not. Some theories are more naturally described in the current framework. For example, as mentioned, the global state of quantum field theory is not prepared by an operation and is more suitably viewed as encoding the correlation of operations. Another example is the process matrices that allow correlations with indefinite causal structure \cite{oreshkov2012quantum, araujo2015witnessing, oreshkov2016causal}. It is found that the process matrices cannot be parallel-composed without constraints \cite{jia2018tensor}. This would appear unnatural if the process matrices are viewed as operations, but natural if they are viewed as correlations among operations.

Another difference lies in the graphical representation of using hypergraphs instead of graphs. Graphical reasoning had been important in previous axiomatic works and works on operational theories in general (see, e.g., \cite{chiribella2010probabilistic, chiribella2011informational, hardy2011reformulating, hardy2013theory}, and \cite{coecke2017picturing} and reference therein). If one chooses to work with the current framework, the natural pictorial tool is the hypergraph, rather than the graph, which is widely used in other models (e.g., \cite{chiribella2009theoretical, chiribella2011informational, hardy2011reformulating, chiribella2010probabilistic, oreshkov2016operational}). Roughly speaking a hypergraph is a generalized graph that allows edges to connect to other integer numbers of nodes rather than just two. The generalized edge is called a ``hyperedge''. We can associate the nodes of a hypergraph to operations/outcomes and the hyperedges to the correlations, connecting the nodes they correlate. The implications of using hypergraphs instead of graphs for probabilistic theories remains to be explored.

\section{The complex Hilbert space structure}\label{sec:chs}

In this section we write down a list of postulates and show that they single out the complex Hilbert space structure. 

We restrict attention to operations with finite dimensional operational spaces. Technically, the reason is that the derivation of the complex Hilbert space structure below uses dimension counting arguments and lemmas that work for finite dimensional spaces. Conceptually, the restriction to work with finite dimensions can be motivated by the constraints of realistic data gathering. Even for theories whose mathematical description uses infinite dimensional spaces such as quantum field theory, realistic data gathering subject to the constraints of finite resolution and finite range, which imply a finite data set. 
Despite these motivations for working with finite dimensional spaces, we do hope that some future work finds a derivation of the complex Hilbert space structure without restricting to finite dimensional spaces. There are useful theories described with infinite dimensional spaces (such as quantum field theory) which introduce new features absent in theories with finite dimensional spaces. It is an open question to what extent the following derivation generalizes to infinite dimensions.

\subsection{Postulates}

To state the next postulate, we need to define the notion of transformation. In ordinary quantum theory, a transformation is a trace non-increasing\footnote{By allowing not just trace-perserving maps we keep the notion of transformation more general. This general notion of transformation is used, for example, in \cite{chiribella2010probabilistic}.} and completely positive map. The trace non-increasing property is required so that absolute probabilities remain in the interval $[0,1]$. The completely positive property is required to ensure that physical states get mapped to physical states even if the transformation acts partially on a subsystem. We want a generalized definition of transformations that applies to all the theories within the current framework. Since the framework uses probability weights instead of absolute probabilities, there is no requirement of the kind of the trace non-increasing property. The following can be viewed as a generalization of complete positive maps.

Fix two arbitrary operational spaces $\mathfrak{O}_{\mathsf{a}}$ and $\mathfrak{O}_{\mathsf{b}}$. We want to define the notion of $\mathsf{a}$-to-$\mathsf{b}$ transformation, which not only maps from $\mathfrak{O}_{\mathsf{a}}$ to $\mathfrak{O}_{\mathsf{b}}$, but also from $\mathfrak{O}_{\mathsf{ac}}$ to $\mathfrak{O}_{\mathsf{bc}}$ for arbitrary $\mathsf{c}$. An \textit{$\mathsf{a}$-to-$\mathsf{b}$ transformation}, denoted by $\mathsf{T_{a,b}}$, is a family $\{T_\mathsf{ac,bc}\}_\mathsf{c}$ of linear maps for each $\mathsf{c}$
\begin{align}
T_\mathsf{ac,bc}: \mathfrak{O}_{\mathsf{ac}}\rightarrow \mathfrak{O}_{\mathsf{bc}},
\end{align}
so that: i) For arbitrary $\mathsf{A}_{\mathsf{a}}\in \mathfrak{O}_{\mathsf{a}}$ and $\mathsf{B}_{\mathsf{c}}\in \mathfrak{O}_{\mathsf{c}}$, $T_\mathsf{ac,bc}(\mathsf{A}_{\mathsf{a}}\mathsf{B}_{\mathsf{c}})= T_\mathsf{a,b}(\mathsf{A}_{\mathsf{a}})\mathsf{B}_{\mathsf{c}}$ for $T_\mathsf{a,b}: \mathfrak{O}_{\mathsf{a}}\rightarrow \mathfrak{O}_{\mathsf{b}}$, and ii) $T_\mathsf{ac,bc}(\mathfrak{O}_{\mathsf{ac}}^+)\subset \mathfrak{O}_{\mathsf{bc}}^+$. Condition i) ensure that the transformation acts locally on product elements and condition ii) generalizes complete positivity.

The transformations as linear maps can be summed linearly. Given $\mathsf{T_{a,b}}=\{T_\mathsf{ac,bc}\}_\mathsf{c}$ and $\mathsf{S_{a,b}}=\{S_\mathsf{ac,bc}\}_\mathsf{c}$, define $p\mathsf{T_{a,b}}+q\mathsf{S_{a,b}}=\{pT_\mathsf{ac,bc}+qS_\mathsf{ac,bc}\}_\mathsf{c}$ for $p,q\in\mathbb{R}$. In this way a vector space is generated. As can be checked easily, the set of all transformations $\mathsf{T_{a,b}}$ forms a convex cone, making the vector space an ordered vector space. Call it a \textit{transformation space} and denote it by $\mathfrak{T}_{\mathsf{a,b}}$. Denote the positive cone by $\mathfrak{T}_{\mathsf{a,b}}^+$ and $\dim\mathfrak{T}_{\mathsf{a,b}}$ by $t_{\mathsf{a,b}}$.

The above definition of transformations is intended as a mathematical characterization of the in principle possible physical transformations. Whether all these mathematically defined transformations are actually realizable, and what the physical interpretation is for the transformations are subject to further specifications of particular theories.\footnote{A commonly seen understanding of a transformation is that it takes states from a previous time to a latter time. This understanding is not general enough. For example, a quantum comb type transformation can transform an operation (which may be a transformation rather than a state) to another operation that extends from an earlier time to a latter time \cite{chiribella2009theoretical}.}

We can now state the postulate.
\begin{postulate}[Dimension]\label{ps:to}
An operational space whose physical system has two subsystems has the same dimension as the correlation space over these two systems, and as the transformation spaces between these two systems.
\end{postulate}
Equivalently, Postulate \ref{ps:to} says that for arbitrary $\mathfrak{O}_{\mathsf{a}}$ and $\mathfrak{O}_{\mathsf{b}}$, $d_{\mathsf{ab}}=c_{\mathsf{ab}}=t_{\mathsf{a,b}}=t_{\mathsf{b,a}}$ (recall that $d_{\mathsf{ab}}=\mathfrak{O}_{\mathsf{ab}}$). One can interpret the postulate as offering the operations enough degrees of freedom to potentially realize all two system correlations and mathematically possible transformations. The correlations of two operations include both those arising from naturally and those controlled by agents. The latter type of correlation must interact with the two relevant systems, and is controlled by the agents through some operations containing the two systems as subsystems. The postulate says that as far as the degrees of freedom of the vector spaces are concerned, the operations have as many degrees of freedom as the set of all possible correlations, including the type arising from nature. Similarly, there are transformations arising from nature and transformations controlled by agents. The agent-controlled transformations between two systems are realized by the agents through some operations pertaining to the two systems as subsystems. The postulate says that as far as the degrees of freedom of the vector spaces are concerned, the operations have as many degrees of freedom as the set of all possible transformations, including the type arising from nature.

We move on from discussing operational space elements transform into each other to how they correlate with each other. Without further constraints the framework allows weird theories such as one in which data recorded from any two operations on different systems are not correlated. In a universe described by this theory little inference can be made. To focus attention on more reasonable theories a postulate is needed to offer some regularity in terms of how systems correlate with each other. We adopt the following ``pairing'' postulate for this purpose.

To state the postulate, first we need the notion of a ``copy'' of operational spaces. An \textit{order-isomorphism} \(f\) between ordered vector spaces \(V\) and \(W\) is a positive, invertible linear map having a positive inverse, where positive means
\(f(V^+) \subseteq W^+\). If two operational spaces $\mathfrak{O}_{\mathsf{a}}$ and $\mathfrak{O}_{\mathsf{b}}$ share an order-isomorphism, we say that they are \textit{copies} of each other. We use primes on physical systems and vectors to signify copies (e.g., $\mathfrak{O}_{\mathsf{a'}}$ for the copy of $\mathfrak{O}_{\mathsf{a}}$, and $\mathsf{A'}_{\mathsf{a'}}$ for the ``copy'' of $\mathsf{A}_{\mathsf{a}}$ under the order-isomorphism).

An operational space $\mathfrak{O}_{\mathsf{a}}$ is said to have a \textit{pairing} if there is a copy $\mathfrak{O}_{\mathsf{a}'}$ and a correlation $\mathsf{C}^{\mathsf{aa'}}$ on the two spaces so that $\mathsf{C}^{\mathsf{aa'}}(\mathsf{A_a},\mathsf{A'_{a'}})>0$ for all nonzero $\mathsf{A_a}\in \mathfrak{O}_{\mathsf{a}}$. The pairing is said to be \textit{symmetric} if $\mathsf{C}^{\mathsf{aa'}}(\mathsf{A_a},\mathsf{B'_{a'}})=\mathsf{C}^{\mathsf{aa'}}(\mathsf{B_a},\mathsf{A'_{a'}})$ for all $\mathsf{A_a, B_a}\in \mathfrak{O}_{\mathsf{a}}$. The pairing is said to be \textit{distinguishing} if whenever an operational space element yields only physical (non-negative) probability weights through the correlation, the element is physical, i.e., whenever $\mathsf{A_a}$ is such that $\mathsf{C}^{\mathsf{aa'}}(\mathsf{A_a},\mathsf{B'_{a'}})\ge 0$ for all $\mathsf{B'_{a'}}\in \mathfrak{O}_{\mathsf{a}'}^+$, $\mathsf{A_{a}}\in \mathfrak{O}_{\mathsf{a}}^+$. A \textit{factorizably symmetric distinguishing pairing} is such that it factorizes for operational spaces with factors while preserving the symmetric and distinguishing properties, i.e., for $\mathfrak{O}_\mathsf{a}=\mathfrak{O}_\mathsf{a_1a_2}$, $\mathsf{A_a}=\mathsf{A_{a_1}}\mathsf{A_{a_2}}$, and $\mathsf{B_a}=\mathsf{B_{a_1}}\mathsf{B_{a_2}}$, $\mathsf{C}^{\mathsf{aa'}}(\mathsf{A_a},\mathsf{B'_{a'}})=\mathsf{C}_1^{\mathsf{a_1a_1'}}(\mathsf{A_{a_1}},\mathsf{B'_{a'_1}})\mathsf{C}_2^{\mathsf{a_2a_2'}}(\mathsf{A_{a_2}},\mathsf{B'_{a'_2}})$ factorizes into two pairings $\mathsf{C}_1^{\mathsf{a_1a_1'}}$ and $\mathsf{C}_2^{\mathsf{a_2a_2'}}$ such that both are symmetric and distinguishing.

\begin{postulate}[Pairing]\label{ps:p}
Each operational space has at least one factorizably symmetric distinguishing pairing.
\end{postulate}

One can interpret the postulate as imposing some regularity on how recorded data correlate. The existence of a pairing offers the possibility to establish some positive correlations for pairs of data recorded with operations, in particular for operations conducted on isomorphic operational spaces, the most elementary pair of spaces that positive correlations can be expected on. The strongest form of correlation we can hope for is that from the outcomes of one operation we can infer unambiguously the outcomes of the paired operation. Postulate \ref{ps:p} is a weaker requirement only asking that paired outcomes appear together with some positive chance (Note that the physical outcomes are elements of the positive cone, so strictly speaking the pairing condition is an extension of the above requirement to all elements of the operational spaces.). The symmetric property appears as a natural assumption for operational spaces that are isomorphic. The distinguishing property assumes that the correlation of the pairing is strong enough to reflect (at the mathematical level) any unphysical correlation if there is any. Finally, the factorizing property is a natural assumption considering the factor structure.

The next postulate is easy to state. An ordered vector space \(V\) is \textit{homogeneous} if \(\text{Aut}(V)\), the group of order-automorphisms on \(V\), acts transitively on the interior of \(V_+\).

\begin{postulate}[Homogeneity]\label{ps:hm}
Operational spaces are homogeneous.
\end{postulate}

Intuitively, the postulate says that inside an operational space any region looks locally like any other. For example, the qubit space of ordinary quantum theory is homogeneous, as there is no preferred direction or region inside the space.

The previous postulates already offer strong constraints to arrive at self-dual (Theorem \ref{th:sd}) and homogeneous spaces, so that only the self-adjoint parts of real,
complex, quaternionic, 3-by-3 octonions matrix algebras, spin factors, and their direct sums are allowed \cite{koecher_geodattischen_1958, vinberg1961homogeneous, jordan_algebraic_1934}. At this stage, a most general theory fulfilling the postulates appears to be direct sum of the different types of the systems listed above. However, in fact as long as a single quantum qubit shows up in the combination, the theory must be exclusively complex Hilbert space quantum (see the Barnum-Wilce Theorem below). The only possibility against this is that a qubit does not show up. Therefore to arrive at the complex Hilbert space we assume:
\begin{postulate}[Qubit]\label{ps:qubit}
There exists a qubit.
\end{postulate}

\subsection{Derivation}

The derivation of the complex Hilbert space structure is simplified immensely thanks to the previous works of Barnum and Wilce  \cite{barnum_local_2014}, Koecher \cite{koecher_geodattischen_1958}, Vinberg \cite{vinberg1961homogeneous}, and Jordan, von Neumann and Wigner \cite{jordan_algebraic_1934}. The relevance of these results is condensed in the Barnum-Wilce Theorem below, which directly yields the final result we look for. To connect the above postulates to the Barnum-Wilce Theorem, we only need to do two simple proofs (Theorem \ref{th:sd} and Theorem \ref{th:tl}).

A finite-dimensional ordered vector space \(V\) is \textit{self-dual} if it has an inner product such that \(a\) belongs to the positive cone \(V^+\) iff \(\langle a,b\rangle\geq 0\) for all \(b\in V_+\).
\begin{theorem}\label{th:sd}
All $\mathfrak{O}_\mathsf{a}$ are self-dual.
\end{theorem}
\begin{proof}
According to Postulate \ref{ps:p}, there is a symmetric distinguishing pairing  $(\mathfrak{O}_{\mathsf{a}'},\mathsf{C^{aa'}})$ for $\mathfrak{O}_\mathsf{a}$. We claim that $\langle \cdot ,\cdot \rangle:\mathfrak{O}_\mathsf{a}\times \mathfrak{O}_\mathsf{a}\rightarrow \mathbb{R}$ defined by $\langle \mathsf{A_a,B_a} \rangle=\mathsf{C^{aa'}}(\mathsf{A_a},\mathsf{B'_{a'}})$ is an inner-product, i.e., it is bilinear, symmetric, and positive definite. The first property follows from Postulate \ref{ps:l}, and the rest from $(\mathfrak{O}_{\mathsf{a}'},\mathsf{C^{aa'}})$ being a symmetric pairing.

Now we show that $\mathsf{A_a}\in\mathfrak{O}^+_{\mathsf{a}}$ iff $\langle \mathsf{A_a,B_a} \rangle\ge 0$ for all $\mathsf{B_a}\in\mathfrak{O}^+_{\mathsf{a}}$. If $\mathsf{A_a}\in\mathfrak{O}^+_{\mathsf{a}}$, then $\langle \mathsf{A_a,B_a} \rangle=\mathsf{C^{aa'}}(\mathsf{A_a},\mathsf{B'_{a'}})\ge 0$ for all $\mathsf{B_a}\in\mathfrak{O}^+_{\mathsf{a}}$ because $\mathsf{C^{aa'}}$ is positive according to Postulate \ref{ps:l}. If $\langle \mathsf{A_a,B_a} \rangle \geq 0$ for all $\mathsf{B_a}\in\mathfrak{O}^+_{\mathsf{a}}$, $\mathsf{A_a}\in\mathfrak{O}^+_{\mathsf{a}}$ by the distinguishing property of the pairing.
\end{proof}

\begin{theorem}[Tomographic locality]\label{th:tl}
$d_{\mathsf{ab}}=d_{\mathsf{a}}d_{\mathsf{b}}=c_{\mathsf{ab}}=c_{\mathsf{a}}c_{\mathsf{b}}$.
\end{theorem}
\begin{proof}
Let \(\mathfrak{O}_\mathsf{a}\) and \(\mathfrak{O}_\mathsf{b}\) be arbitrary. We want to count the number of degrees of freedom, $t_{\mathsf{a,b}}$, in defining a transformation \(\mathsf{T_{b,a}}\in\mathfrak{T}_{\mathsf{b,a}}\). These degrees of freedom fix its action on all \(\mathsf{A}_{\mathsf{ac}}\in \mathfrak{O}_{\mathsf{ac}}\) for arbitrary \(\mathsf{c}\). First let \(\mathsf{c}\) be trivial. The local action of \(\mathsf{T_{ab}}\) on $\mathfrak{O}_\mathsf{a}$ is determined by taking \(d_{\mathsf{a}}\) linearly independent vectors from $\mathfrak{O}_\mathsf{a}$ and specifying an image for each. Each image requires \(d_{\mathsf{b}}\) real parameters to specify, so \(d_{\mathsf{a}}d_{\mathsf{b}}\) independent real parameters are needed in total.

Now let $\mathsf{c}=\mathsf{b}$. Condition ii) in the definition of transformations fixes the action of \(\mathsf{T_{ab}}\) on product elements of the form $\mathsf{A_a B_b}$, but the action on the possible additional linearly independent elements is yet unspecified. For each of the \(r_{\mathsf{ab}}:=d_{\mathsf{ab}}-d_{\mathsf{a}}d_{\mathsf{b}}\geq 0\) additional linearly independent vectors, \(d_{\mathsf{bb}}\) real parameters are needed to determine the image. Hence specifying \(\mathsf{T_{ab}}\) requires at least \(l_{\mathsf{ab}}:=d_{\mathsf{a}}d_{\mathsf{b}}+r_{\mathsf{ab}}d_{\mathsf{bb}}\) independent real parameters, i.e., \(t_{\mathsf{b,a}} \geq l_{\mathsf{ab}}\). By Postulate \ref{ps:to}, $t_{\mathsf{b,a}}=d_{\mathsf{ab}}$, so
\begin{align}\label{eq:di}
l_{\mathsf{ab}}-t_{\mathsf{b,a}}=d_{\mathsf{a}}d_{\mathsf{b}}+r_{\mathsf{ab}}d_{\mathsf{bb}}-d_{\mathsf{ab}}=r_{\mathsf{ab}}(d_{\mathsf{bb}}-1)\leq 0.
\end{align}
If \(d_{\mathsf{b}}>1\), \(d_{\mathsf{bb}}\geq d_{\mathsf{b}}^2>1\). By (\ref{eq:di}), \(r_{\mathsf{ab}}=0\). If otherwise \(d_{\mathsf{b}}=1\), then \(r_{\mathsf{ab}}=d_{\mathsf{ab}}-d_{\mathsf{a}}d_{\mathsf{b}}=d_{\mathsf{a}}-d_{\mathsf{a}}=0\). Therefore \(r_{\mathsf{ab}}\) is always $0$, and \(d_{\mathsf{ab}}=d_{\mathsf{a}}d_{\mathsf{b}}\). By Postulate \ref{ps:to}, this also equals to $c_{\mathsf{ab}}$. Taking $\mathsf{b}$ to be trivial, we see that $c_{\mathsf{a}}=d_{\mathsf{a}}$. Therefore $d_{\mathsf{ab}}=d_{\mathsf{a}}d_{\mathsf{b}}=c_{\mathsf{ab}}=c_{\mathsf{a}}c_{\mathsf{b}}$.
\end{proof}

In Proposition 1.1 of \cite{barnum_local_2014}, Barnum and Wilce proved the following result.
\begin{theorem*}[Barnum-Wilce]
	For a homogeneous and factorizably self-dual probabilistic theory, if it obeys tomographic locality and contains a qubit, then all its systems are self-adjoint parts of complex matrix algebras.
\end{theorem*}

The theorem was originally obtained in the context of no-signaling probabilistic theories with definite causal structure. However, the proof of the theorem also goes through in the present context as allowing indefinite causal structure does not affect the proof and no-signaling was only used to prove that maps of the form $(\mathsf{A_a},\mathsf{B_b})\mapsto \mathsf{C^{ab}(A_a, B_b)}$ are bilinear, which holds automatically in our framework. In the theorem, \textit{factorizably self-dual} means that the self-dualizing inner product factors on two subsystems, i.e., \(\langle \mathsf{A_a B_b, X_a Y_b} \rangle=\langle \mathsf{A_a , X_a} \rangle \langle \mathsf{B_b , Y_b} \rangle\). This property holds for the self-dualizing product in Theorem \ref{th:sd} if we pick the pairing to be factorizable, as allowed by Postulate \ref{ps:p}. This plus Postulates \ref{ps:hm}, \ref{ps:qubit} and Theorem \ref{th:tl} leads to the following result.
\begin{corollary*}\label{cor:qt}
All operational space elements are self-adjoint parts of complex matrices.
\end{corollary*}

\section{Conclusion}\label{sec:con}

We presented a general framework for physical theories that does not assume definite causal structure. This framework takes operations and correlations as the central concepts. We further identified a list of postulates from which finite dimensional complex Hilbert space quantum theories can be derived. This may be viewed as an axiomatic formulation of quantum theories without assuming indefinite causal structure.

More than one quantum theory is compatible with the postulates. The compatible theories include both quantum theories with explicit indefinite causal structure (e.g.,  \cite{chiribella2013quantum, oreshkov2012quantum, oreshkov2015operational, oreshkov2016operational, silva2017connecting}), and ordinary formulations of quantum theory with definite causal structure (definite causality can be imposed as a further postulate). This leads to the interesting question if one among these many compatible theories describes nature best.

The framework presented in Section \ref{sec:oc} allows infinite dimensional systems\footnote{Finite dimensional spaces are used to motivate Postulate \ref{ps:l}, but as stated the postulate does not impose finite dimensionality.}, and can in principle incorporate infinite dimensional theories such as quantum field theory. It is an interesting open question to identify postulates that derive infinite dimensional quantum theory without assuming definite causal structure.

\begin{acknowledgements}
The author thanks Lucien Hardy and Achim Kempf for guidance and support, Alexander Wilce and Matthew Leifer for making suggestions and pointing out flaws in an earlier version of the project, and Ognyan Oreshkov, Carlo Maria Scandolo, Lee Smolin, and Matthew Graydon for discussion.

Research at Perimeter Institute is supported by the Government of Canada through the Department of Innovation, Science and Economic Development Canada and by the Province of Ontario through the Ministry of Research, Innovation and Science.  This work is partly supported by a grant from the John Templeton Foundation. The opinions expressed in this work are those of the author's and do not necessarily reflect the views of the John Templeton Foundation.
\end{acknowledgements}

\bibliographystyle{apsrev}
\bibliography{newer}

\begin{thebibliography}{43}
\expandafter\ifx\csname natexlab\endcsname\relax\def\natexlab#1{#1}\fi
\expandafter\ifx\csname bibnamefont\endcsname\relax
  \def\bibnamefont#1{#1}\fi
\expandafter\ifx\csname bibfnamefont\endcsname\relax
  \def\bibfnamefont#1{#1}\fi
\expandafter\ifx\csname citenamefont\endcsname\relax
  \def\citenamefont#1{#1}\fi
\expandafter\ifx\csname url\endcsname\relax
  \def\url#1{\texttt{#1}}\fi
\expandafter\ifx\csname urlprefix\endcsname\relax\def\urlprefix{URL }\fi
\providecommand{\bibinfo}[2]{#2}
\providecommand{\eprint}[2][]{\url{#2}}

\bibitem[{\citenamefont{MacLean et~al.}(2017)\citenamefont{MacLean, Ried,
  Spekkens, and Resch}}]{maclean2017quantum}
\bibinfo{author}{\bibfnamefont{J.-P.~W.} \bibnamefont{MacLean}},
  \bibinfo{author}{\bibfnamefont{K.}~\bibnamefont{Ried}},
  \bibinfo{author}{\bibfnamefont{R.~W.} \bibnamefont{Spekkens}},
  \bibnamefont{and} \bibinfo{author}{\bibfnamefont{K.~J.} \bibnamefont{Resch}},
  \bibinfo{journal}{Nature Communications} \textbf{\bibinfo{volume}{8}}
  (\bibinfo{year}{2017}).

\bibitem[{\citenamefont{Procopio et~al.}(2015)\citenamefont{Procopio, Moqanaki,
  Ara{\'u}jo, Costa, Calafell, Dowd, Hamel, Rozema, Brukner, and
  Walther}}]{procopio2015experimental}
\bibinfo{author}{\bibfnamefont{L.~M.} \bibnamefont{Procopio}},
  \bibinfo{author}{\bibfnamefont{A.}~\bibnamefont{Moqanaki}},
  \bibinfo{author}{\bibfnamefont{M.}~\bibnamefont{Ara{\'u}jo}},
  \bibinfo{author}{\bibfnamefont{F.}~\bibnamefont{Costa}},
  \bibinfo{author}{\bibfnamefont{I.~A.} \bibnamefont{Calafell}},
  \bibinfo{author}{\bibfnamefont{E.~G.} \bibnamefont{Dowd}},
  \bibinfo{author}{\bibfnamefont{D.~R.} \bibnamefont{Hamel}},
  \bibinfo{author}{\bibfnamefont{L.~A.} \bibnamefont{Rozema}},
  \bibinfo{author}{\bibfnamefont{{\v{C}}.}~\bibnamefont{Brukner}},
  \bibnamefont{and} \bibinfo{author}{\bibfnamefont{P.}~\bibnamefont{Walther}},
  \bibinfo{journal}{Nature Communications} \textbf{\bibinfo{volume}{6}},
  \bibinfo{pages}{7913} (\bibinfo{year}{2015}).

\bibitem[{\citenamefont{Rubino et~al.}(2017)\citenamefont{Rubino, Rozema, Feix,
  Ara{\'u}jo, Zeuner, Procopio, Brukner, and Walther}}]{rubino2017experimental}
\bibinfo{author}{\bibfnamefont{G.}~\bibnamefont{Rubino}},
  \bibinfo{author}{\bibfnamefont{L.~A.} \bibnamefont{Rozema}},
  \bibinfo{author}{\bibfnamefont{A.}~\bibnamefont{Feix}},
  \bibinfo{author}{\bibfnamefont{M.}~\bibnamefont{Ara{\'u}jo}},
  \bibinfo{author}{\bibfnamefont{J.~M.} \bibnamefont{Zeuner}},
  \bibinfo{author}{\bibfnamefont{L.~M.} \bibnamefont{Procopio}},
  \bibinfo{author}{\bibfnamefont{{\v{C}}.}~\bibnamefont{Brukner}},
  \bibnamefont{and} \bibinfo{author}{\bibfnamefont{P.}~\bibnamefont{Walther}},
  \bibinfo{journal}{Science Advances} \textbf{\bibinfo{volume}{3}},
  \bibinfo{pages}{e1602589} (\bibinfo{year}{2017}).

\bibitem[{\citenamefont{Hardy}(2009)}]{hardy2009quantum}
\bibinfo{author}{\bibfnamefont{L.}~\bibnamefont{Hardy}}, in
  \emph{\bibinfo{booktitle}{Quantum Reality, Relativistic Causality, and
  Closing the Epistemic Circle}} (\bibinfo{publisher}{Springer},
  \bibinfo{year}{2009}), pp. \bibinfo{pages}{379--401}.

\bibitem[{\citenamefont{Chiribella et~al.}(2013)\citenamefont{Chiribella,
  D’Ariano, Perinotti, and Valiron}}]{chiribella2013quantum}
\bibinfo{author}{\bibfnamefont{G.}~\bibnamefont{Chiribella}},
  \bibinfo{author}{\bibfnamefont{G.~M.} \bibnamefont{D’Ariano}},
  \bibinfo{author}{\bibfnamefont{P.}~\bibnamefont{Perinotti}},
  \bibnamefont{and} \bibinfo{author}{\bibfnamefont{B.}~\bibnamefont{Valiron}},
  \bibinfo{journal}{Physical Review A} \textbf{\bibinfo{volume}{88}},
  \bibinfo{pages}{022318} (\bibinfo{year}{2013}).

\bibitem[{\citenamefont{Chiribella}(2012)}]{chiribella2012perfect}
\bibinfo{author}{\bibfnamefont{G.}~\bibnamefont{Chiribella}},
  \bibinfo{journal}{Physical Review A} \textbf{\bibinfo{volume}{86}},
  \bibinfo{pages}{040301} (\bibinfo{year}{2012}).

\bibitem[{\citenamefont{Ara{\'u}jo et~al.}(2014)\citenamefont{Ara{\'u}jo,
  Costa, and Brukner}}]{araujo2014computational}
\bibinfo{author}{\bibfnamefont{M.}~\bibnamefont{Ara{\'u}jo}},
  \bibinfo{author}{\bibfnamefont{F.}~\bibnamefont{Costa}}, \bibnamefont{and}
  \bibinfo{author}{\bibfnamefont{{\v{C}}.}~\bibnamefont{Brukner}},
  \bibinfo{journal}{Physical Review Letters} \textbf{\bibinfo{volume}{113}},
  \bibinfo{pages}{250402} (\bibinfo{year}{2014}).

\bibitem[{\citenamefont{Feix et~al.}(2015)\citenamefont{Feix, Ara{\'u}jo, and
  Brukner}}]{feix2015quantum}
\bibinfo{author}{\bibfnamefont{A.}~\bibnamefont{Feix}},
  \bibinfo{author}{\bibfnamefont{M.}~\bibnamefont{Ara{\'u}jo}},
  \bibnamefont{and}
  \bibinfo{author}{\bibfnamefont{{\v{C}}.}~\bibnamefont{Brukner}},
  \bibinfo{journal}{Physical Review A} \textbf{\bibinfo{volume}{92}},
  \bibinfo{pages}{052326} (\bibinfo{year}{2015}).

\bibitem[{\citenamefont{Gu{\'e}rin et~al.}(2016)\citenamefont{Gu{\'e}rin, Feix,
  Ara{\'u}jo, and Brukner}}]{guerin2016exponential}
\bibinfo{author}{\bibfnamefont{P.~A.} \bibnamefont{Gu{\'e}rin}},
  \bibinfo{author}{\bibfnamefont{A.}~\bibnamefont{Feix}},
  \bibinfo{author}{\bibfnamefont{M.}~\bibnamefont{Ara{\'u}jo}},
  \bibnamefont{and}
  \bibinfo{author}{\bibfnamefont{{\v{C}}.}~\bibnamefont{Brukner}},
  \bibinfo{journal}{Physical review letters} \textbf{\bibinfo{volume}{117}},
  \bibinfo{pages}{100502} (\bibinfo{year}{2016}).

\bibitem[{\citenamefont{Nielsen and Chuang}(2000)}]{nielsen2000quantum}
\bibinfo{author}{\bibfnamefont{M.~A.} \bibnamefont{Nielsen}} \bibnamefont{and}
  \bibinfo{author}{\bibfnamefont{I.~L.} \bibnamefont{Chuang}},
  \emph{\bibinfo{title}{Quantum Computation and Quantum Information}}
  (\bibinfo{publisher}{Cambridge University Press}, \bibinfo{year}{2000}).

\bibitem[{\citenamefont{Hardy}({\natexlab{a}})}]{hardy2005probability}
\bibinfo{author}{\bibfnamefont{L.}~\bibnamefont{Hardy}},
  \eprint{arXiv:gr-qc/0509120}.

\bibitem[{\citenamefont{Hardy}(2007)}]{hardy2007towards}
\bibinfo{author}{\bibfnamefont{L.}~\bibnamefont{Hardy}},
  \bibinfo{journal}{Journal of Physics A: Mathematical and Theoretical}
  \textbf{\bibinfo{volume}{40}}, \bibinfo{pages}{3081} (\bibinfo{year}{2007}).

\bibitem[{\citenamefont{Oreshkov et~al.}(2012)\citenamefont{Oreshkov, Costa,
  and Brukner}}]{oreshkov2012quantum}
\bibinfo{author}{\bibfnamefont{O.}~\bibnamefont{Oreshkov}},
  \bibinfo{author}{\bibfnamefont{F.}~\bibnamefont{Costa}}, \bibnamefont{and}
  \bibinfo{author}{\bibfnamefont{{\v{C}}.}~\bibnamefont{Brukner}},
  \bibinfo{journal}{Nature Communications} \textbf{\bibinfo{volume}{3}},
  \bibinfo{pages}{1092} (\bibinfo{year}{2012}).

\bibitem[{\citenamefont{Oreshkov and Cerf}(2016)}]{oreshkov2016operational}
\bibinfo{author}{\bibfnamefont{O.}~\bibnamefont{Oreshkov}} \bibnamefont{and}
  \bibinfo{author}{\bibfnamefont{N.~J.} \bibnamefont{Cerf}},
  \bibinfo{journal}{New Journal of Physics} \textbf{\bibinfo{volume}{18}},
  \bibinfo{pages}{073037} (\bibinfo{year}{2016}).

\bibitem[{\citenamefont{Stueckelberg}(1960)}]{stueckelberg1960quantum}
\bibinfo{author}{\bibfnamefont{E.~C.} \bibnamefont{Stueckelberg}},
  \bibinfo{journal}{Helv. Phys. Acta} \textbf{\bibinfo{volume}{33}},
  \bibinfo{pages}{458} (\bibinfo{year}{1960}).

\bibitem[{\citenamefont{Auletta}(2001)}]{auletta2001foundations}
\bibinfo{author}{\bibfnamefont{G.}~\bibnamefont{Auletta}},
  \emph{\bibinfo{title}{Foundations and Interpretation of Quantum Mechanics: In
  the Light of a Critical-Historical Analysis of the Problems and of a
  Synthesis of the Results}} (\bibinfo{publisher}{World Scientific},
  \bibinfo{year}{2001}).

\bibitem[{\citenamefont{Chiribella and Spekkens}(2016)}]{chiribella2016quantum}
\bibinfo{editor}{\bibfnamefont{G.}~\bibnamefont{Chiribella}} \bibnamefont{and}
  \bibinfo{editor}{\bibfnamefont{R.~W.} \bibnamefont{Spekkens}}, eds.,
  \emph{\bibinfo{title}{Quantum Theory: Informational Foundations and Foils}}
  (\bibinfo{publisher}{Springer}, \bibinfo{year}{2016}).

\bibitem[{\citenamefont{Popescu and Rohrlich}(1994)}]{popescu1994quantum}
\bibinfo{author}{\bibfnamefont{S.}~\bibnamefont{Popescu}} \bibnamefont{and}
  \bibinfo{author}{\bibfnamefont{D.}~\bibnamefont{Rohrlich}},
  \bibinfo{journal}{Foundations of Physics} \textbf{\bibinfo{volume}{24}},
  \bibinfo{pages}{379} (\bibinfo{year}{1994}).

\bibitem[{\citenamefont{Barrett}(2007)}]{barrett2007information}
\bibinfo{author}{\bibfnamefont{J.}~\bibnamefont{Barrett}},
  \bibinfo{journal}{Physical Review A} \textbf{\bibinfo{volume}{75}},
  \bibinfo{pages}{032304} (\bibinfo{year}{2007}).

\bibitem[{\citenamefont{Hardy}({\natexlab{b}})}]{hardy2001quantum}
\bibinfo{author}{\bibfnamefont{L.}~\bibnamefont{Hardy}},
  \eprint{arXiv:quant-ph/0101012}.

\bibitem[{\citenamefont{Barnum et~al.}(2014)\citenamefont{Barnum, Markus,
  Ududec et~al.}}]{barnum_higher-order_2014}
\bibinfo{author}{\bibfnamefont{H.}~\bibnamefont{Barnum}},
  \bibinfo{author}{\bibfnamefont{P.}~\bibnamefont{Markus}},
  \bibinfo{author}{\bibfnamefont{C.}~\bibnamefont{Ududec}},
  \bibnamefont{et~al.}, \bibinfo{journal}{New Journal of Physics}
  \textbf{\bibinfo{volume}{16}}, \bibinfo{pages}{123029}
  (\bibinfo{year}{2014}).

\bibitem[{\citenamefont{Chiribella et~al.}(2011)\citenamefont{Chiribella,
  D’Ariano, and Perinotti}}]{chiribella2011informational}
\bibinfo{author}{\bibfnamefont{G.}~\bibnamefont{Chiribella}},
  \bibinfo{author}{\bibfnamefont{G.~M.} \bibnamefont{D’Ariano}},
  \bibnamefont{and}
  \bibinfo{author}{\bibfnamefont{P.}~\bibnamefont{Perinotti}},
  \bibinfo{journal}{Physical Review A} \textbf{\bibinfo{volume}{84}},
  \bibinfo{pages}{012311} (\bibinfo{year}{2011}).

\bibitem[{\citenamefont{Daki{\'c} and Brukner}(2011)}]{daki?_quantum_2011}
\bibinfo{author}{\bibfnamefont{B.}~\bibnamefont{Daki{\'c}}} \bibnamefont{and}
  \bibinfo{author}{\bibfnamefont{{\v C}.}~\bibnamefont{Brukner}}, in
  \emph{\bibinfo{booktitle}{Deep Beauty: Understanding the Quantum World
  through Mathematical Innovation}}, edited by
  \bibinfo{editor}{\bibfnamefont{H.}~\bibnamefont{Halvorson}}
  (\bibinfo{publisher}{Cambridge University Press}, \bibinfo{year}{2011}), pp.
  \bibinfo{pages}{365--392}.

\bibitem[{\citenamefont{Hardy}({\natexlab{c}})}]{hardy2011reformulating}
\bibinfo{author}{\bibfnamefont{L.}~\bibnamefont{Hardy}},
  \eprint{arXiv:1104.2066}.

\bibitem[{\citenamefont{Masanes and
  M{\"u}ller}(2011)}]{masanes_derivation_2011}
\bibinfo{author}{\bibfnamefont{L.}~\bibnamefont{Masanes}} \bibnamefont{and}
  \bibinfo{author}{\bibfnamefont{M.~P.} \bibnamefont{M{\"u}ller}},
  \bibinfo{journal}{New Journal of Physics} \textbf{\bibinfo{volume}{13}},
  \bibinfo{pages}{063001} (\bibinfo{year}{2011}).

\bibitem[{\citenamefont{Wilce}(2012)}]{wilce2012four}
\bibinfo{author}{\bibfnamefont{A.}~\bibnamefont{Wilce}}, in
  \emph{\bibinfo{booktitle}{Probability in Physics}}, edited by
  \bibinfo{editor}{\bibfnamefont{Y.}~\bibnamefont{Ben-Menahem}}
  \bibnamefont{and} \bibinfo{editor}{\bibfnamefont{M.}~\bibnamefont{Hemmo}}
  (\bibinfo{publisher}{Springer}, \bibinfo{year}{2012}), pp.
  \bibinfo{pages}{281--298}.

\bibitem[{\citenamefont{Wilce}()}]{wilce_conjugates_2012}
\bibinfo{author}{\bibfnamefont{A.}~\bibnamefont{Wilce}},
  \eprint{arXiv:1206.2897}.

\bibitem[{\citenamefont{Selby et~al.}()\citenamefont{Selby, Scandolo, and
  Coecke}}]{selby2018reconstructing}
\bibinfo{author}{\bibfnamefont{J.~H.} \bibnamefont{Selby}},
  \bibinfo{author}{\bibfnamefont{C.~M.} \bibnamefont{Scandolo}},
  \bibnamefont{and} \bibinfo{author}{\bibfnamefont{B.}~\bibnamefont{Coecke}},
  \eprint{arXiv:1802.00367}.

\bibitem[{\citenamefont{Barnum and Wilce}(2016)}]{barnum2016post}
\bibinfo{author}{\bibfnamefont{H.}~\bibnamefont{Barnum}} \bibnamefont{and}
  \bibinfo{author}{\bibfnamefont{A.}~\bibnamefont{Wilce}}, in
  \emph{\bibinfo{booktitle}{Quantum Theory: Informational Foundations and
  Foils}}, edited by
  \bibinfo{editor}{\bibfnamefont{G.}~\bibnamefont{Chiribella}}
  \bibnamefont{and} \bibinfo{editor}{\bibfnamefont{R.~W.}
  \bibnamefont{Spekkens}} (\bibinfo{publisher}{Springer},
  \bibinfo{year}{2016}), pp. \bibinfo{pages}{367--420}.

\bibitem[{\citenamefont{Davies and Lewis}(1970)}]{davies1970operational}
\bibinfo{author}{\bibfnamefont{E.~B.} \bibnamefont{Davies}} \bibnamefont{and}
  \bibinfo{author}{\bibfnamefont{J.~T.} \bibnamefont{Lewis}},
  \bibinfo{journal}{Communications in Mathematical Physics}
  \textbf{\bibinfo{volume}{17}}, \bibinfo{pages}{239} (\bibinfo{year}{1970}).

\bibitem[{\citenamefont{Oreshkov and Cerf}(2015)}]{oreshkov2015operational}
\bibinfo{author}{\bibfnamefont{O.}~\bibnamefont{Oreshkov}} \bibnamefont{and}
  \bibinfo{author}{\bibfnamefont{N.~J.} \bibnamefont{Cerf}},
  \bibinfo{journal}{Nature Physics} \textbf{\bibinfo{volume}{11}},
  \bibinfo{pages}{853} (\bibinfo{year}{2015}).

\bibitem[{\citenamefont{Chiribella et~al.}(2010)\citenamefont{Chiribella,
  D’Ariano, and Perinotti}}]{chiribella2010probabilistic}
\bibinfo{author}{\bibfnamefont{G.}~\bibnamefont{Chiribella}},
  \bibinfo{author}{\bibfnamefont{G.~M.} \bibnamefont{D’Ariano}},
  \bibnamefont{and}
  \bibinfo{author}{\bibfnamefont{P.}~\bibnamefont{Perinotti}},
  \bibinfo{journal}{Physical Review A} \textbf{\bibinfo{volume}{81}},
  \bibinfo{pages}{062348} (\bibinfo{year}{2010}).

\bibitem[{\citenamefont{Ara{\'u}jo et~al.}(2015)\citenamefont{Ara{\'u}jo,
  Branciard, Costa, Feix, Giarmatzi, and Brukner}}]{araujo2015witnessing}
\bibinfo{author}{\bibfnamefont{M.}~\bibnamefont{Ara{\'u}jo}},
  \bibinfo{author}{\bibfnamefont{C.}~\bibnamefont{Branciard}},
  \bibinfo{author}{\bibfnamefont{F.}~\bibnamefont{Costa}},
  \bibinfo{author}{\bibfnamefont{A.}~\bibnamefont{Feix}},
  \bibinfo{author}{\bibfnamefont{C.}~\bibnamefont{Giarmatzi}},
  \bibnamefont{and}
  \bibinfo{author}{\bibfnamefont{{\v{C}}.}~\bibnamefont{Brukner}},
  \bibinfo{journal}{New Journal of Physics} \textbf{\bibinfo{volume}{17}},
  \bibinfo{pages}{102001} (\bibinfo{year}{2015}).

\bibitem[{\citenamefont{Oreshkov and Giarmatzi}(2016)}]{oreshkov2016causal}
\bibinfo{author}{\bibfnamefont{O.}~\bibnamefont{Oreshkov}} \bibnamefont{and}
  \bibinfo{author}{\bibfnamefont{C.}~\bibnamefont{Giarmatzi}},
  \bibinfo{journal}{New Journal of Physics} \textbf{\bibinfo{volume}{18}},
  \bibinfo{pages}{093020} (\bibinfo{year}{2016}).

\bibitem[{\citenamefont{Jia and Sakharwade}(2018)}]{jia2018tensor}
\bibinfo{author}{\bibfnamefont{D.}~\bibnamefont{Jia}} \bibnamefont{and}
  \bibinfo{author}{\bibfnamefont{N.}~\bibnamefont{Sakharwade}},
  \bibinfo{journal}{Physical Review A} \textbf{\bibinfo{volume}{97}},
  \bibinfo{pages}{032110} (\bibinfo{year}{2018}).

\bibitem[{\citenamefont{Hardy}(2013)}]{hardy2013theory}
\bibinfo{author}{\bibfnamefont{L.}~\bibnamefont{Hardy}}, in
  \emph{\bibinfo{booktitle}{Computation, Logic, Games, and Quantum Foundations.
  The Many Facets of Samson Abramsky}} (\bibinfo{publisher}{Springer},
  \bibinfo{year}{2013}), pp. \bibinfo{pages}{83--106}.

\bibitem[{\citenamefont{Coecke and Kissinger}(2017)}]{coecke2017picturing}
\bibinfo{author}{\bibfnamefont{B.}~\bibnamefont{Coecke}} \bibnamefont{and}
  \bibinfo{author}{\bibfnamefont{A.}~\bibnamefont{Kissinger}},
  \emph{\bibinfo{title}{Picturing quantum processes}}
  (\bibinfo{publisher}{Cambridge University Press}, \bibinfo{year}{2017}).

\bibitem[{\citenamefont{Chiribella et~al.}(2009)\citenamefont{Chiribella,
  D’Ariano, and Perinotti}}]{chiribella2009theoretical}
\bibinfo{author}{\bibfnamefont{G.}~\bibnamefont{Chiribella}},
  \bibinfo{author}{\bibfnamefont{G.~M.} \bibnamefont{D’Ariano}},
  \bibnamefont{and}
  \bibinfo{author}{\bibfnamefont{P.}~\bibnamefont{Perinotti}},
  \bibinfo{journal}{Physical Review A} \textbf{\bibinfo{volume}{80}},
  \bibinfo{pages}{022339} (\bibinfo{year}{2009}).

\bibitem[{\citenamefont{Koecher}(1958)}]{koecher_geodattischen_1958}
\bibinfo{author}{\bibfnamefont{M.}~\bibnamefont{Koecher}},
  \bibinfo{journal}{Mathematische Annalen} \textbf{\bibinfo{volume}{135}},
  \bibinfo{pages}{192} (\bibinfo{year}{1958}).

\bibitem[{\citenamefont{Vinberg}(1961)}]{vinberg1961homogeneous}
\bibinfo{author}{\bibfnamefont{E.}~\bibnamefont{Vinberg}},
  \bibinfo{journal}{Dokl. Acad. Nauk. SSSR} \textbf{\bibinfo{volume}{141}},
  \bibinfo{pages}{270} (\bibinfo{year}{1961}), \bibinfo{note}{(English trans.
  Soviet Math. Dokl 2, 1416 (1961))}.

\bibitem[{\citenamefont{Jordan et~al.}(1934)\citenamefont{Jordan, Neumann, and
  Wigner}}]{jordan_algebraic_1934}
\bibinfo{author}{\bibfnamefont{P.}~\bibnamefont{Jordan}},
  \bibinfo{author}{\bibfnamefont{J.~v.} \bibnamefont{Neumann}},
  \bibnamefont{and} \bibinfo{author}{\bibfnamefont{E.}~\bibnamefont{Wigner}},
  \bibinfo{journal}{Annals of Mathematics} \textbf{\bibinfo{volume}{35}},
  \bibinfo{pages}{29} (\bibinfo{year}{1934}).

\bibitem[{\citenamefont{Barnum and Wilce}(2014)}]{barnum_local_2014}
\bibinfo{author}{\bibfnamefont{H.}~\bibnamefont{Barnum}} \bibnamefont{and}
  \bibinfo{author}{\bibfnamefont{A.}~\bibnamefont{Wilce}},
  \bibinfo{journal}{Foundations of Physics} \textbf{\bibinfo{volume}{44}},
  \bibinfo{pages}{192} (\bibinfo{year}{2014}).

\bibitem[{\citenamefont{Silva et~al.}(2017)\citenamefont{Silva, Guryanova,
  Short, Skrzypczyk, Brunner, and Popescu}}]{silva2017connecting}
\bibinfo{author}{\bibfnamefont{R.}~\bibnamefont{Silva}},
  \bibinfo{author}{\bibfnamefont{Y.}~\bibnamefont{Guryanova}},
  \bibinfo{author}{\bibfnamefont{A.~J.} \bibnamefont{Short}},
  \bibinfo{author}{\bibfnamefont{P.}~\bibnamefont{Skrzypczyk}},
  \bibinfo{author}{\bibfnamefont{N.}~\bibnamefont{Brunner}}, \bibnamefont{and}
  \bibinfo{author}{\bibfnamefont{S.}~\bibnamefont{Popescu}},
  \bibinfo{journal}{New Journal of Physics} \textbf{\bibinfo{volume}{19}},
  \bibinfo{pages}{103022} (\bibinfo{year}{2017}).

\end{thebibliography}

\end{document}